\documentclass[pra,twocolumn,nofootinbib]{revtex4-2}
\pdfoutput=1

\usepackage{graphicx}      
\usepackage[section]{placeins}
\usepackage{amsmath,amssymb,amsthm} 
\usepackage{dsfont}       
\usepackage{bbm}         
\usepackage{physics}      
\usepackage{color}   
\usepackage{braket} 
\usepackage{dutchcal}
\usepackage{comment}
\usepackage{mathtools}
\usepackage{thmtools}

\usepackage[colorlinks=true,linkcolor=blue,citecolor=blue,urlcolor=blue]{hyperref}
\newcommand\id{\mathbbm{1}}

\newcommand{\ack}{\subsection*{\normalsize \sf 
		\textbf{Acknowledgements}}}

\newtheorem{theorem}{Theorem}
\newtheorem{lemma}{Lemma}
\newtheorem*{lemma*}{Lemma}
\newtheorem{definition}{Definition}

\newtheorem{corollary}{Corollary}

\newtheorem*{theorem*}{Theorem}

\begin{document}
	
	\title{No-go theorem for quantum realization of extremal correlations}
	
	\author{Sujan V.K}
	\email{sujan.v\textunderscore k@students.iiserpune.ac.in}
	\affiliation{Indian Institute of Science Education and Research, Pune 411008, India}
	
	\author{Ravi Kunjwal}
	\email{quaintum.research@gmail.com}
	\affiliation{Universit\'e libre de Bruxelles, QuIC, Brussels, Belgium}
	\affiliation{Aix-Marseille University, CNRS, LIS, Marseille, France}
	
	\begin{abstract}
		The study of quantum correlations is central to quantum information and foundations. The paradigmatic case of Bell scenarios considers product measurements implemented on a multipartite state. The more general case of contextuality scenarios---where the measurements do not have to be of product form or even on a composite system---has been studied for the case of projective measurements. While it is known that in any Bell scenario extremal indeterministic correlations (\textit{e.g.}, Popescu-Rohrlich or PR boxes) are unachievable quantumly, the case of general contextuality scenarios has remained open. Here we study quantum realizations of extremal correlations in arbitrary contextuality scenarios and prove that, for all such scenarios, no extremal indeterministic correlation can be achieved using projective quantum measurements, \textit{i.e.}, there exists no quantum state and no set of projective measurements, for any contextuality scenario, that can achieve such correlations. 
		This no-go result follows as a corollary of a more general no-go theorem that holds when the most general set of quantum measurements (\textit{i.e.}, positive operator-valued measures, or POVMs) is taken into account. This general no-go theorem entails that no non-trivial quantum realization of an extremal indeterministic correlation exists, \textit{i.e.}, any ``quantum" realization must be simulable by classical randomness. We discuss implications of this no-go theorem and the open questions it raises.
	\end{abstract}
	\maketitle
	
\section{Introduction}
    The intrinsic indeterminism of quantum theory is key to its ability to provide quantum advantage in any scenario, whether physical,  informational, or computational. This intrinsic indeterminism follows from the Born rule which, in turn, can be obtained via the Busch-Gleason theorem \cite{Busch03,Gleason57} once we fix the structure of quantum measurements to be positive operator-valued measures (POVMs) and require some simple constraints on how probabilities are assigned to the outcomes of these measurements. Bell's theorem \cite{Bell64,Bell66,Bell76,Bellbook} and the Kochen-Specker theorem \cite{KS67} provide the strongest witnesses of such indeterminism using only a finite set of measurements.
	
	This raises a natural question: how far does the intrinsic indeterminism of quantum theory go? More precisely, while we know that a deterministic account of quantum correlations is impossible under assumptions like Bell-locality or KS-noncontextuality, what do we know about how extremal the intrinsic indeterminism of quantum theory can be? Here, by `extremal' we mean convexly extremal, \textit{i.e.}, a correlation which is extremal cannot be decomposed into a probabilistic mixture of other correlations. Bell-local correlations that are deterministic---that is, vertices of the Bell polytope---are an example of extremal \textit{deterministic} correlations. On the other hand, Bell-nonlocal vertices of the nonsignaling polytope in Bell scenarios correspond to extremal \textit{indeterministic} correlations. In Ref.~\cite{RTH16} it was shown that in Bell scenarios the intrinsic indeterminism of quantum theory cannot be extremal, \textit{i.e.}, Bell-nonlocal vertices of the nonsignaling polytope are unachievable via quantum strategies. 
	
	Bell scenarios correspond to a very special kind of set-up: a multipartite state shared across $N$ $(\geq 2)$ parties that can each implement local measurements on their share of the state but where no party can signal to any other party. The nonsignaling assumption in Bell scenarios is physically enforced by an appeal to relativistic causality, \textit{i.e.}, by ensuring spacelike separation between the $N$ parties during the course of their local measurements. It is possible to abstract away from this physical picture to a combinatorial one by representing Bell scenarios using hypergraphs called \textit{contextuality scenarios} \cite{AFL15}, also known as \textit{test spaces} in quantum logic \cite{Wilce09,Greechie71,FR72,FR73}.\footnote{Hypergraphs are a generalization of graphs: while each edge in a graph connects two vertices, each hyperedge in a hypergraph can connect an arbitrary number of vertices.} In the case of Bell scenarios, the vertices in such a hypergraph correspond to multipartite measurement outcomes and the hyperedges correspond to the normalization and nonsignaling conditions satisfied by the probabilities that can be assigned to these outcomes. However, contextuality scenarios capture far more general set-ups than those of  Bell scenarios, \textit{i.e.}, scenarios where, on a fixed state, many different measurements can be carried out. In particular, contextuality scenarios do not assume that the state is multipartite or that the measurements are local with respect to multipartite states. The vertices in this general case represent measurement outcomes and the hyperedges represent measurement settings. When a vertex appears in two hyperedges, it represents the fact that it occurs with the same probability for either measurement setting, regardless of the state on which these measurements are implemented, \textit{e.g.}, a projection that appears in two different  projection-valued measures (PVMs, or projective measurements).
	
	In this article we consider arbitrary contextuality scenarios and ask:\\
	
	\textit{Is a quantum realization of an extremal indeterministic correlation possible in any contextuality scenario?}\\
	
	We show that such a realization is  impossible if we require the measurements to be projective and, in the general case of positive operator-valued measures (POVMs), no quantum realization can be non-trivial, \textit{i.e.}, it is always possible to simulate it by discarding the quantum system at hand and using a classical source of randomness. While the Busch-Gleason \cite{Busch03,Gleason57} and Kochen-Specker (KS) theorems \cite{KS67} show that no quantum state can yield extremal deterministic correlations for \textit{all} quantum measurements, our no-go theorem (\autoref{thm:4}) shows that no quantum state can yield extremal indeterministic correlations for \textit{any} set of quantum measurements. In the special case of Bell scenarios, we recover the main result of Ref.~\cite{RTH16}, namely, no Bell-nonlocal vertices of the nonsignaling polyope in any Bell scenario are realizable with quantum measurements.
	
	It is quantum information folklore---often termed the ``Church of the Larger Hilbert space"---that, fundamentally, one need not worry about the POVMs, mixed states, and non-unitary dynamics of operational quantum theory because they can, without loss of generality, be replaced by PVMs, pure states, and unitary evolutions as long as one can freely extend the system's Hilbert space to include sufficiently large auxiliary systems. The mathematical basis of this folklore are the dilation theorems of Stinespring and Naimark. This reasoning works well in Bell scenarios---where POVMs can always be replaced by PVMs in this manner without ever violating the nonsignaling constraints---but it falters in general contextuality scenarios where replacing POVMs by PVMs is not always possible without doing violence to the assumptions underlying the contextuality scenario. We discuss this apparent ``failure" of Naimark dilations in view of our no-go theorem, pointing to examples in the literature. Related limitations of Naimark's dilation theorem were also discussed in Ref.~\cite{HFR14}.
	
	\section{Preliminaries}
We define below the relevant notions that we will need to prove our results, largely drawing on the terminology of Ref.~\cite{AFL15}.

\begin{definition}[Contextuality scenario]\label{def:1}
	A \emph{contextuality scenario} is a hypergraph $H$ with set of vertices $V(H)$ and set of edges $E(H) \subseteq 2^{V(H)}$ such that
	
	\begin{align}
		\bigcup_{e \in E(H)} e = V(H).
	\end{align}
	
\end{definition}

\begin{definition}[Probabilistic model]\label{def:2}
	Let $H$ be a contextuality scenario. A \emph{probabilistic model} on $H$ is an assignment $p : V(H) \to [0,1] $ of a probability $p(v)$ to each vertex $v \in V(H)$ such that 
	
	\begin{align}
		\sum_{v \in e} p(v) = 1 \quad \forall e \in E(H).
	\end{align}
	We denote the set of (general) probabilistic models on $H$ by $\mathcal{G}(H)$. This set defines a convex polytope and we refer to the vertices of this polytope as \emph{extremal probabilistic models}.
\end{definition}

\begin{definition}[Induced subscenario]\label{def:indsub}
	Given a contextuality scenario $H$ and a subset of its vertices $S\subseteq V(H)$, $S$ induces a subscenario of $H$---the induced subscenario $H_S$---which is a contextuality scenario  with
	\begin{align}
		V(H_S):=S,\quad E(H_S):=\{e\cap S: e\in E(H)\}.
	\end{align}
\end{definition}

\begin{definition}[Extension of a probabilistic model]\label{def:extension}
	Given a contextuality scenario $H$ and its subscenario $H_S$ induced by $S\subseteq V(H)$, every probabilistic model $p_S\in\mathcal{G}(H_S)$ defines a probabilistic model $p\in\mathcal{G}(H)$, called the \textit{extension of $p_S$}, given by
	\begin{align}
		p(v):=\begin{cases}
			p_S(v) &\emph{ if }v\in S,\\
			0 &\emph{ if }v\notin S.
		\end{cases}
	\end{align}
\end{definition}

\begin{definition}[Incidence matrix of a contextuality scenario]\label{def:incmat}
	The incidence matrix $A_H$ of a contextuality scenario $H$ is a $|E(H)|\times |V(H)|$ matrix with entries in $\{0,1\}$ given by 
	\begin{align}
		[A_H]_{(e,v)}:=\begin{cases}
			1\emph{ if }v\in e,\\
			0\emph{ if }v\notin e,
		\end{cases}
	\end{align}
	where $v\in V(H)$ and $e\in E(H)$.
\end{definition}
Note that for any $p\in\mathcal{G}(H)$, understood as a $|V(H)|\times 1$ column vector, we have that $A_H p=\textbf{1}$, \textit{i.e.}, the incidence matrix captures the normalization constraints of $H$ and thus provides an equivalent representation of $H$. Here $\textbf{1}$ represents the $|V(H)|\times 1$ column vector with all entries $1$. We will also use $\textbf{0}$ to represent the $|V(H)|\times 1$ column vector with all entries $0$.

Below we recall a theorem characterizing extremal probabilistic models on any contextuality scenario from Ref.~\cite{AFL15} that will be crucial for our main result.
\begin{theorem}[Theorem 2.5.3 in \cite{AFL15}]\label{thm:1}
	A probabilistic model $p$ on a contextuality scenario $H$, \textit{i.e.}, $p \in \mathcal{G}(H)$, is extremal if and only if it is the extension of $p_S \in \mathcal{G}(H_S)$, where $H_S$ is some induced subscenario of H which has $p_S$ as its unique probabilistic model (that is, $\mathcal{G}(H_S)=\{p_S\}$).
\end{theorem}

Probabilistic models on contextuality scenarios are interesting from the perspective of quantum theory because they allow us to represent classical and quantum correlations within the same combinatorial framework. Specifically, following Ref.~\cite{AFL15}, one can define the set of classical and quantum models as follows.

\begin{definition}[Classical models]\label{def:classmod}
	Given a contextuality scenario $H$, $p\in\mathcal{G}(H)$ is said to be a \emph{classical model} if it admits a convex decomposition into deterministic models $p_k\in\mathcal{G}(H)$, where $p_k(v)\in\{0,1\}$ for all $v\in V(H)$ and for all $k$,  such that 
	\begin{align}
		\forall v\in V(H): p(v)=\sum_k q_k p_k(v),
	\end{align}
	where $q_k\geq 0$ for all $k$ and $\sum_kq_k=1$. We denote the set of classical models on $H$ by $\mathcal{C}(\mathcal{H})$.
\end{definition}
In terming these models ``classical" models, we are following the terminology of Ref.~\cite{AFL15}. It is more precise, however, to refer to these models as \textit{KS-noncontextual models} since KS-noncontextuality is the precise notion of classicality that their definition captures \cite{Kunjwal16,Kunjwal19,WK23}.

\begin{definition}[Quantum models]\label{def:3}
	Given a contextuality scenario $H$, $p\in\mathcal{G}(H)$ is a \emph{quantum model} if there exists a Hilbert space $\mathcal{H}$, a quantum state $\rho \in \mathcal{B}(\mathcal{H})$, and a projection operator  
	$P_v \in \mathcal{B}(\mathcal{H})$ associated to every $v \in V(H)$, such that every hyperedge $e\in E(H)$ forms a projective measurement, \textit{i.e.},
	\begin{align}
		\sum_{v \in e} P_v = \id\quad \forall e \in E(H),
	\end{align}
	and this construction reproduces the probabilistic model $p$, \textit{i.e.},
	\begin{align}
		p(v) = \mathrm{tr}(\rho P_v) \quad \forall v \in V(H).
	\end{align}
	We denote the set of quantum models on $H$ by $\mathcal{Q}(H)$, where $\mathcal{C}(H)\subseteq\mathcal{Q}(H)\subseteq\mathcal{G}(H)$.
\end{definition}

We now define the notion of quantum realization of a  probabilistic model (which may not be a quantum model in the sense of \autoref{def:3}).

\begin{definition}[Quantum realization of a probabilistic model]\label{def:4}
	Given a contextuality scenario $H$, a \emph{quantum realization} of $p\in\mathcal{G}(H)$ is given by a Hilbert space $\mathcal{H}$, a quantum state $\rho\in\mathcal{B}(\mathcal{H})$, and a set of positive semidefinite operators $ \{E_v\in\mathcal{B}(\mathcal{H})\}_{v\in V(H)}$, \textit{i.e.},
	\begin{align}
		(\mathcal{H},\rho, \{E_v\}_{v\in V(H)}),
	\end{align}
	such that 
	\begin{align}
		\sum_{v\in e} E_v= \id\quad \forall\, e \in E, \textrm{ and }
	\end{align}
	\begin{align}
		p(v)= \mathrm{tr}\!\left(\rho\, E_v\right) \quad \forall v\in V(H).
	\end{align}
\end{definition}
Clearly, every quantum model (by definition) admits a quantum realization where the positive operators are projectors. However, the converse is not true, \textit{i.e.}, there exist probabilistic models $p\in\mathcal{G}(H)\backslash\mathcal{Q}(H)$ that admit a quantum realization. In fact, \textit{all} probabilistic models $p\in\mathcal{G}(H)$ admit the following trivial quantum realization:
\begin{align}
	(\mathcal{H},\rho,\{p(v)\id\}_{v\in V(H)}),
\end{align}
for arbitrary choices of $\mathcal{H}$ (with dimension $\geq 1$) and $\rho$. However, there also exist non-trivial quantum realizations of probabilistic models that are not in $\mathcal{Q}(H)$, \textit{e.g.}, in the case of Specker's scenario \cite{Specker60,Specker60_1,Specker60_2,GKS18}, where one can come up with non-trivial quantum realizations that exhibit quantum correlations stronger than what is possible with projective measurements, such as the violation of the Liang-Spekkens-Wiseman (LSW) inequality \cite{LSW11,KG14,Kunjwal15,ZCL17}. Below we formalize the notion of a trivial quantum realization:

\begin{definition}[Trivial quantum realization of a probabilistic model]\label{def:trivialq}
	Given a contextuality scenario $H$, a \emph{quantum realization} $(\mathcal{H},\rho, \{E_v\}_{v\in V(H)})$ of $p\in\mathcal{G}(H)$ is said to be trivial if it takes the following general form:
	\begin{align}
		\rho&:=\Lambda_r \oplus 0_{d-r},\\
		E_v&:=p(v)\id_r \oplus E^{(d-r)}_v\quad\forall v\in V(H),
	\end{align}
	where $r$ is the rank of $\rho$, $d$ is the dimension of $\mathcal{H}$, and in the eigenbasis of $\rho$ we have $\Lambda_r:=\emph{diag}(\lambda_1,\lambda_2,\dots,\lambda_r)$, $0_{d-r}:=\emph{diag}(0,0,\dots,0)$, $\id_r:=\emph{diag}(1,1,\dots,1)$, and $E_v^{(d-r)}$ are positive operators subject only to the normalization constraints $\sum_{v\in e}E_v^{(d-r)}=\id_{d-r}$ for all $e\in E(H)$.
\end{definition}
It should be clear why we call such a quantum realization trivial: any probabilistic model $p\in\mathcal{G}(H)$ that can be reproduced in the above manner on a Hilbert space of dimension $d$ can also be reproduced by the following quantum realization on a Hilbert space of dimension $1\leq r\leq d$:
\begin{align}
	(\mathbb{C}^r,\Lambda_r,\{p(v)\id_r\}_{v\in V(H)}),
\end{align}
which is equivalent to a measurement device just ignoring the quantum system at hand and generating the probabilistic model using a classical source of randomness that is \textit{independent} of the quantum state.\footnote{It suffices to imagine that the quantum measurement is secretly a classical random number generator that is pre-programmed, for each setting $e\in E(H)$, to produce the outcome $v\in e$ with probability $p(v)$---independent of $e$, as required by \autoref{def:2}---\textit{regardless} of the particular quantum state fed to it. Such a ``measurement" is completely noisy, \textit{i.e.}, we can learn nothing about the quantum state by implementing it. Note that here we use the word ``classical" in a different sense---namely, ``implementable using only classical resources"---than in \autoref{def:classmod}, where it has a more constrained meaning, namely, ``implementable using only classical resources in a KS-noncontextual manner". In the language of the ontological models framework \cite{HS10}, the former would be a `measurement noncontextual' ontological model and the latter would be a `KS-noncontextual' ontological model \cite{Spekkens05,LSW11,Kunjwal15,Kunjwal16}. Quantum theory respects measurement noncontextuality but contradicts KS-noncontextuality.}
	
	\section{Contextuality scenarios}
	
	\subsection{Key lemmas}
	We first state some key lemmas about contextuality scenarios that will allow us to prove our main result. We refer the reader to Appendix \ref{app:proofs} for proofs of these lemmas.
	\begin{lemma}\label{lem:zero}
		Given a quantum state $\rho$ and an effect $E$ on a Hilbert space of dimension $d$ such that $\Tr(\rho E)=0$, they admit the following general form:
		\begin{align}
			\rho = \begin{pmatrix}
				\Lambda_r & 0 \\
				0 & 0
			\end{pmatrix}, E=\begin{pmatrix}
				0 & 0 \\
				0 & E_{22}
			\end{pmatrix},
		\end{align}
		where $r$ is the rank of $\rho$, $\Lambda_r = \operatorname{diag}(\lambda_1, \ldots, \lambda_r)$ with $\lambda_k > 0$ for $k\in\{1,2,\dots,r\}$, and $E_{22}$ is a $(d-r)\times(d-r)$ positive semidefinite matrix.
	\end{lemma}
	
	\begin{lemma}\label{lem:trivial}
		A positive semidefinite operator $E\geq 0$ on a Hilbert space $\mathcal{H}$ satisfies 
			$\Tr(\rho E)=\alpha$ for all quantum states $\rho\in\mathcal{B}(\mathcal{H})$ 
		if and only if $E=\alpha\id$.
	\end{lemma}
	
	\begin{lemma}\label{lem:unique}
		A contextuality scenario $H$ that admits the unique probabilistic model $p$ (\textit{i.e.}, $\mathcal{G}(H)=\{p\}$) admits no other quantum realization besides the following family of trivial quantum realizations: 
		\begin{align}
			(\mathcal{H},\rho,\{E_v:=p(v)\id\}_{v\in V(H)}),
		\end{align}
		where $\dim\mathcal{H} \ge 1$ and the choice of state $\rho$ on $\mathcal{H}$ is arbitrary.
	\end{lemma}
	
	\begin{lemma}\label{lem:incidence}
		The equation $Ax = \textbf{0}$, where $A$ is the incidence matrix of a contextuality scenario $H$ with a unique probabilistic model $p$ satisfying $p(v)>0$ for all $v\in V(H)$, admits only the trivial solution, \textit{i.e.}, $x=\textbf{0}$.
	\end{lemma}

	\subsection{No-go theorem for non-trivial quantum realizations}
	With these lemmas in hand, we can now prove our main theorem:
	\begin{theorem}[No-go theorem for non-trivial quantum realizations]\label{thm:4}
		For any contextuality scenario $H$, quantum realizations of its extremal probabilistic models in $\mathcal{G}(H)$ must necessarily be trivial.
	\end{theorem}
	\begin{proof}
		Consider a quantum realization $$(\mathcal{H},\rho, \{E_v\}_{v\in V(H)})$$ of some extremal probabilistic
		model $p$ on some contextuality scenario $H$. We have $p(v) = \operatorname{Tr}(\rho E_v)$
		and without loss of generality, we can express $\rho$ and $\{E_v\}_{v\in V(H)}$ in the eigenbasis of $\rho$, \textit{i.e.},
		\begin{align}
			\rho = \begin{pmatrix}
				\Lambda_r & 0 \\
				0 & 0
			\end{pmatrix},
		\end{align}
		where $\Lambda_r$ is a diagonal positive definite $r \times r$ matrix  (and $r$=rank(\(\rho\))).
		
		For any measurement operator $E_v$ corresponding to $p(v)=0$ in the extremal model, we have $\operatorname{Tr}(\rho E_v) = 0$. This implies that $E_v$ must be orthogonal to the support of $\rho$ (\textit{cf.}~\autoref{lem:zero}):
		\begin{align}
			E_v= \begin{pmatrix}
				0 & 0 \\
				0 & E^{(v)}_{22}
			\end{pmatrix}   
		\end{align}
		
		Now note that, from \autoref{thm:1}, the subscenario induced by $p$ on $H$ (since $p$ is extremal)---say, $H_S$---admits the unique probabilistic model---say $p_S$---of which $p$ is an extension. Hence, positive operators corresponding to vertices with non-zero probabilities in the extremal model $p$ define the unique probabilistic model $p_S$ on the induced subscenario $H_S$ with vertices $S:=\{v\in V(H)| p(v)>0\}\subseteq V(H)$. Applying \autoref{lem:unique} to the support of $\rho$ (the upper-diagonal block), we conclude that on this subspace all the operators $\{E_v\}_{v\in S}$ must be proportional to $\id_r$. Specifically, if $E_v$ corresponds to probability $p(v)> 0$ in the extremal model, then
		\begin{align}
			E_v = \begin{pmatrix}
				p(v)\id_r & E^{(v)}_{12} \\
				E^{(v)}_{21} & E^{(v)}_{22}
			\end{pmatrix},
		\end{align}
		where $\id_r$ is the identity matrix on $\text{supp}(\rho)$.
		
		The off-diagonal blocks $E^{(v)}_{12}$ and $E^{(v)}_{21}$ are constrained by the requirement that for each hyperedge $e\in E(H)$, the completeness relation must hold:
		\begin{align}
			\sum_{v\in e} E^{(v)}_{12} = \sum_{v\in e} E^{(v)}_{21}=0.
		\end{align}
		Let us collect all the $(i,j)$-th entries of all the matrices $\{E^{(v)}_{12}\}_{v\in S}$ in a vector $\textbf{x}_{ij}:=([E^{(v)}_{12}]_{ij})_{v\in S}$, where $i\in\{1,2,\dots,r\},j\in\{r+1,r+2,\dots,d\}$. Then we must have 
		\begin{align}
			\forall (i,j): A_{H_S}\textbf{x}_{ij}=\textbf{0},
		\end{align}
		where $A_{H_S}$ is the incidence matrix of $H_S$. Since $H_S$ admits a unique probabilistic model, we have from \autoref{lem:incidence} that $\textbf{x}_{ij}=\textbf{0}$ for all $(i,j)$ and, therefore, $E^{(v)}_{12}=0$ for all $v\in S$. Since $E^{(v)}_{21}=E^{(v)\dagger}_{12}$, we also have that $E^{(v)}_{21}=0$ for all $v\in S$.
		
		The lower-diagonal blocks $E^{(v)}_{22}$ corresponding to the null space of $\rho$ are unconstrained by the extremal probabilistic model since they do not contribute to the Born rule probabilities. The only requirement is their positive semidefiniteness and the fact that for each hyperedge $e\in E(H)$, the completeness relation must hold, \textit{i.e.},
			$E^{(v)}_{22}\geq 0\,\forall v\in V(H)$ and $\sum_{v\in e} E^{(v)}_{22} = \id_{d-r}\,\forall e\in E(H)$, where $\id_{d-r}$ is the identity matrix on $\text{null}(\rho)$.
		
		Therefore, any quantum realization that generates the extremal probabilistic model $p\in\mathcal{G}(H)$ must be of the form
		\begin{align}
			\left(\mathcal{H}, \Lambda_r \oplus 0_{d-r}, \{p(v) \id_r \oplus E^{(v)}_{22}\}_{v\in V(H)}\right),
		\end{align}
		\textit{i.e.}, such a realization must be trivial.
	\end{proof}
	The following corollary specializes the above theorem to the case of projective measurements:
	\begin{corollary}[No-go for projective quantum realizations]\label{cor:ctx}
		For any contextuality scenario $H$, the set of quantum models $\mathcal{Q}(H)$ excludes extremal indeterministic probabilistic models in $\mathcal{G}(H)$, \textit{i.e.}, it is impossible to realize an extremal probabilistic model that is indeterministic using projective quantum measurements.
	\end{corollary}
	\begin{proof}
		This follows from the fact that the trivial quantum realization of \autoref{thm:4} is projective only if $p(v)\in\{0,1\}$ for all $v\in V(H)$. 
		More explicitly, we have
		\begin{align}
			E^2_v = \begin{pmatrix}
				p(v)^2\id_r & 0 \\
				0 & (E^{(v)}_{22})^2
			\end{pmatrix},
		\end{align}
		so that, for a projective realization, we must have  
		\begin{align}
			E^2_v=E_v\quad\forall v\in V(H),
		\end{align}
		that is,
		\begin{align}
			p(v)^2=p(v), (E^{(v)}_{22})^2=E^{(v)}_{22}\quad\forall v\in V(H).
		\end{align}
		This gives us the $p(v)\in\{0,1\}$ for all $v\in V(H)$ as a necessary condition for a quantum projective realization. Hence, no extremal indeterministic probabilistic model on $H$ admits a projective quantum realization.
	\end{proof}
	
	\subsection{Bell scenarios}
	We can now apply \autoref{thm:4} to the case of Bell scenarios---which form a subset of contextuality scenarios \cite{AFL15}---and provide an alternative, and arguably simpler, proof of the central result of Ref.~\cite{RTH16}, \textit{i.e.},	
	\begin{corollary}\label{cor:bell}
		No Bell-nonlocal vertex of the nonsignaling polytope in any Bell scenario admits a quantum realization.
	\end{corollary}
		\begin{proof}
		Let $\{p(v)\}_{v\in V(H_N)}$ denote an extremal probabilistic model on an $N$-party Bell scenario $H_N$. Here each vertex $v\in V(H_N)$ is defined by a set of $N$-party outcomes that result from a set of $N$-party settings, \textit{i.e.}, we have the association $v\leftrightarrow(\vec{a}|\vec{x})$, where $\vec{x}=(x_1,\ldots,x_N)$ is a vector of measurement settings and $\vec{a}=(a_1,\ldots,a_N)$ is the vector of corresponding outcomes. For party $i\in\{1,2,\dots,N\}$, we denote its settings by $x_i\in\mathcal{X}_i$ and the outcomes for each setting $x_i$ by $a_i\in \mathcal{A}_{x_i}$.
		
		From \autoref{thm:4}, any quantum realization of the extremal probabilistic model must be of the form
		\begin{align*}
			\Bigl(\mathcal{H},\Lambda_r \oplus 0_{d-r}, \;\{\,p(\vec{a}|\vec{x})\id_r \oplus E^{(\vec{a}|\vec{x})}_{22}\,\}_{v\in V(H)}\Bigr),
		\end{align*}
		where $d=\textrm{dim}(\mathcal{H})$. Note that this realization allows us to choose an arbitrary $\Lambda_r$, for any given $r$ and $d$ ($1\leq r\leq d$), since the probabilistic model is independent of this choice. 
		Hence, without loss of generality, we can set
		\begin{align}
			\Lambda_r=\textrm{diag}(1,0,\dots,0),
		\end{align}
		so that the quantum state of interest is 
		\begin{align}
			\rho=\Lambda_r\oplus 0_{d-r}.
		\end{align}
		The constraint of locality of the measurements in a Bell scenario with $N$ parties requires that each POVM element factorizes as
		\begin{align}
			E_{\vec{a}|\vec{x}}
			:= E_{a_1|x_1} \otimes E_{a_2|x_2} \otimes \cdots \otimes E_{a_N|x_N},
		\end{align}
		where each positive operator $E_{a_i|x_i}$ is defined on a local Hilbert space $\mathcal{H}_i$ of dimension $d_i$ ($i\in\{1,2,\dots,N\}$) such that $\prod_{i=1}^Nd_i=d$.
		
		For each party $i\in\{1,2,\dots,N\}$, the operators $\{E_{a_i|x_i}\}_{a_i \in \mathcal{A}_{x_i}}$ form a POVM associated with measurement choice $x_i\in\mathcal{X}_i$, \textit{i.e.},
		\begin{align}
			\sum_{a_i \in \mathcal{A}_{x_i}} E_{a_i|x_i} = \mathbbm{1}
			\qquad \forall\, x_i\in\mathcal{X}_i.
		\end{align}
		Hence, we have that
		\begin{align}
			\sum_{\vec{a} \in \mathcal{A}_{\vec{x}}} 
			E_{\vec{a}\,|\,\vec{x}} = \mathbbm{1}
			\qquad \forall\, \vec{x}\in\mathcal{X}_1\times\mathcal{X}_2\times\dots\times\mathcal{X}_N,    
		\end{align}
		with $\mathcal{A}_{\vec{x}}:=\mathcal{A}_{x_1}\times\cdots\times \mathcal{A}_{x_N}$. Given the particular quantum realization we are working with, we have 
		\begin{align}
			E_{\vec{a}|\vec{x}}=p(\vec{a}|\vec{x})\id_r\oplus E_{22}^{(\vec{a}|
				\vec{x})}.
		\end{align}
		On computing the probabilities, this gives us 
		\begin{align}
			&\Tr(\rho E_{\vec{a}|\vec{x}})=p(\vec{a}|\vec{x}). 
		\end{align}
		Expanding the left-hand-side of this equation, we have 
		\begin{align}
			&\Tr(\rho E_{a_1|x_1} \otimes E_{a_2|x_2} \otimes \cdots \otimes E_{a_N|x_N})\nonumber\\
			&= [E_{a_1|x_1}]_{11} [E_{a_2|x_2}]_{11}\dots[E_{a_N|x_N}]_{11} \nonumber \\
			&= \prod_{i=1}^Np(a_i|x_i),
		\end{align}
		where 
		\begin{align}
			p(a_i|x_i):=\Tr(\rho E_{a_i|x_i}\otimes \id_{d/d_i})=[E_{a_i|x_i}]_{11},\,\forall x_i,a_i,i.
		\end{align}
		
		We must therefore have 
		\begin{align}
			p(\vec{a}|\vec{x}) = p(a_1|x_1)\, p(a_2|x_2) \cdots p(a_N|x_N)
		\end{align}
		satisfying
		\begin{align}
			\sum_{a_i \in \mathcal{A}_{x_i}} p(a_i|x_i) = 1,
			\qquad \forall\, x_i,
		\end{align}
		
		\begin{align}
			\sum_{\vec{a} \in \mathcal{A}_{\vec{x}}} 
			p(\vec{a}\,|\,\vec{x}) = 1,
			\qquad \forall\, \vec{x}.
		\end{align}
		Hence, any extremal probabilistic model that admits a quantum realization factorizes into a product of local probability distributions and must, therefore, be deterministic, \textit{i.e.}, a vertex of the Bell polytope. This rules out the possibility of quantum realizations of Bell-nonlocal vertices of the non-signalling polytope since all these vertices are extremal indeterministic probabilistic models on $H_N$ (which in turn are not factorizable as products of local probability distributions).\\
	\end{proof}
	
	\section{Projective vs.~Non-projective realizations: the case of Naimark dilations}
	
	Naimark's dilation theorem \cite{HFR14} states that we can always view a POVM as the compression of a PVM on a higher-dimensional Hilbert space, \textit{i.e.}, for any POVM $\{E_a\}_a$ on some Hilbert space $\mathcal{H}$, there exists a PVM $\{\Pi_a\}_a$ on a Hilbert space $\mathcal{K}$ (where $\textrm{dim}(\mathcal{K})\geq \textrm{dim}(\mathcal{H})$) such that 
	\begin{align}
		&\forall a: E_a=V^{\dagger}\Pi_a V,\\
		&\textrm{where }V:\mathcal{H}\rightarrow\mathcal{K}\textrm{ is an isometry,}\nonumber\\
		&\textit{i.e., } V^{\dagger}V=\id_{\mathcal{H}}.
	\end{align} 
	Hence, for any given quantum state $\rho$ on $\mathcal{H}$, we have that 
	\begin{align}
		\Tr_{\mathcal{H}}(\rho E_a)=\Tr_{\mathcal{H}}(\rho V^{\dagger}\Pi_a V)=\Tr_{\mathcal{K}}(V\rho V^{\dagger}\Pi_a).
	\end{align}
	
	In a Bell scenario, allowing parties to implement local POVMs does not lead to a more general set of correlations than allowing them to only implement local PVMs as long as no restriction is put on the local Hilbert space dimensions, \textit{i.e.}, both cases lead to the same set of correlations. This is because one can always dilate the local POVMs for each party to local PVMs and given the local isometries that accomplish this, namely, $V_{x_k}:\mathcal{H}_k\rightarrow \mathcal{K}_{k}$ (where  $V^{\dagger}_{x_k} V_{x_k}=\id_{\mathcal{H}_k}$) for all $k\in\{1,2,\dots,N\}$, we can define 
	\begin{align}
		&V_{\vec{x}}:=V_{x_1}\otimes V_{x_2}\otimes\dots\otimes  V_{x_N}: \bigotimes_{k=1}^N\mathcal{H}_k\rightarrow \bigotimes_{k=1}^N\mathcal{K}_{k},\\
		&\sigma_{\vec{x}}:=V_{\vec{x}}\rho V_{\vec{x}}^{\dagger},\quad \Pi_{a_k|x_k}:=V_{x_k} E_{a_k|x_k} V^{\dagger}_{x_k},
	\end{align}
	such that
	\begin{align}
		\forall \vec{a}, \vec{x}:&\Tr(\rho E_{a_1|x_1}\otimes E_{a_2|x_2}\otimes\dots\otimes E_{a_N|x_N})\nonumber\\
		=&\Tr(\sigma_{\vec{x}}\Pi_{a_1|x_1}\otimes\Pi_{a_2|x_2}\otimes\dots\otimes\Pi_{a_N|x_N}).
	\end{align}
	
	The above reasoning does not, however, go through for arbitrary contextuality scenarios, \textit{i.e.}, in this general case the set of probabilistic models possible with POVMs is strictly larger than those possible with PVMs, even without any restriction on the Hilbert space dimension.  This is because in an arbitrary contextuality scenario, unlike a Bell scenario, there is no \textit{a priori} tensor product decomposition of the Hilbert space that is specified by the scenario. In Bell scenarios it is precisely this tensor product decomposition which guarantees that the most general global isometries implementable by the parties---each constrained to act non-trivially only on its local Hilbert space---are given by tensor products of local isometries. This ensures that such isometries cannot lead to a violation of the nonsignaling constraints, \textit{i.e.,} the PVMs they define respect the same nonsignaling constraints that are respected by the starting set of POVMs (of which the PVMs are Naimark dilations). In terms of the hypergraph associated with the Bell scenario, this means that each vertex $v\leftrightarrow(\vec{a}|\vec{x})$ that is assigned some positive operator $E_{v}$ in the POVM realization is assigned a projection $\Pi_v$ that is independent of any particular hyperedge $e\in E(H)$ that contains $v$.
	
	In general contextuality scenarios, by contrast, the isometries that follow from Naimark's dilation theorem can act non-trivially on the full Hilbert space since there is no notion of ``parties" that factors this Hilbert space into a tensor product of local Hilbert spaces. More concretely, we have that for each hyperedge $e\in E(H)$ with associated POVM $\{E_v\}_{v\in e}$, we can define an isometry $V_e$ such that 
	\begin{align}
		\{\Pi_v\}_{v\in e}=\{V_eE_v V_e^{\dagger}\}_{v\in e}
	\end{align} 
	is a PVM. However, since each $v$ may appear in multiple hyperedges, these isometries are subject to additional constraints, namely,
	\begin{align}
		&\forall v\in e\cap e'\textrm{ where }e,e'\in E(H):\nonumber\\
		&\Tr(\sigma V_eE_v V_e^{\dagger})=\Tr(\sigma V_{e'}E_v V_{e'}^{\dagger}), \forall\sigma\in\mathcal{B}(\mathcal{K}),
	\end{align}
	which, in turn, is equivalent to the constraint that 
	\begin{align}
		&\forall v\in e\cap e'\textrm{ where }e,e'\in E(H):\nonumber\\
		&V_eE_v V_e^{\dagger}=V_{e'}E_v V_{e'}^{\dagger}.
	\end{align}
	Naimark's dilation theorem does not guarantee the existence of isometries that satisfy these additional constraints while preserving the probabilistic model. Hence, we could have a situation where the only possible isometries $V_e$ and $V_{e'}$ for distinct hyperedges $e,e'$ containing $v$ are such that
	\begin{align}
		V_eE_v V_e^{\dagger}\neq V_{e'}E_v V_{e'}^{\dagger},
	\end{align} 
	thereby demonstrating the impossibility of achieving a probabilistic model realized by POVMs on some Hilbert space $\mathcal{H}$ by PVMs on some Hilbert space $\mathcal{K}$,\textit{ i.e.}, a probabilistic model that his realizable by POVMs but is not a quantum model in the sense of being realizable by PVMs. \autoref{thm:4} and \autoref{cor:ctx} demonstrate that extremal indeterministic models on a contextuality scenario are examples of such probabilistic models. Non-extremal examples of such probabilistic models can be obtained by considering contextuality scenarios that arise from POVMs whose joint measurability structure \cite{KHF14} is not reproducible by PVMs. We discuss some of these contextuality scenarios below.
	
	\subsection*{Non-extremal probabilistic models admitting a non-trivial quantum realization outside $\mathcal{Q}(H)$}
	To our knowledge, the first example of a non-extremal probabilistic model $p\in\mathcal{G}(H)$ on a contextuality scenario $H$ such that $p\notin\mathcal{Q}(H)$ was obtained in Ref.~\cite{KG14}. This contextuality scenario (\textit{cf.}~Ref.~\cite{Kunjwal15} for a discussion of the probabilistic models on it) is obtained from three dichotomic POVMs that are pairwise compatible but triplewise incompatible, realizing a set of (in)compatibility relations termed ``Specker's scenario" \cite{Specker60,Specker60_1,Specker60_2, LSW11}. Ref.~\cite{KG14} showed a violation of the Liang-Spekkens-Wiseman (LSW) inequality \cite{LSW11,Kunjwal15} using qubit POVMs, and in doing so, it demonstrated the existence of a non-extremal probabilistic model on the corresponding contextuality scenario that is not a quantum model. The fact that such a non-extremal probabilistic model cannot be a quantum model follows from the observation that pairwise compatible PVMs are also triplewise compatible and cannot, therefore, produce correlations outside the set of classical models $\mathcal{C}(H)$ (\textit{cf.}~\autoref{def:classmod}). Another example of a non-extremal probabilistic model that is not a quantum model was demonstrated in Ref.~\cite{GKS18}, albeit no non-trivial quantum realization of it was provided in Ref.~\cite{GKS18}. We conjecture that a non-trivial quantum realization of the probabilistic model in Ref.~\cite{GKS18} is indeed possible.
	
	We mention the above examples from the literature to emphasize that it is not \textit{a priori} ``pathological" to consider probabilistic models that arise from POVMs \cite{HS15}. What is pathological is to apply the notion of KS-noncontextuality to POVMs since that leads to absurdities such as trivial POVMs achieving maximal KS-contextuality in the sense of realizing all extremal probabilistic models (see Appendix C of Ref.~\cite{Kunjwal19} for a discussion of this specific aspect). We refer the interested reader to Refs.~\cite{Spekkens14,KS15,Kunjwal16, KS18,Kunjwal19} for a discussion of why generalized contextuality, rather than KS-contextuality, is the appropriate notion of nonclassicality when considering the most general set of operations (including POVMs) in quantum theory. For the specific case of generalized contextuality tailored to scenarios inspired by KS-contextuality---rather than arbitrary prepare-and-measure scenarios---we refer the interested reader to the frameworks developed in Refs.~\cite{KS15,KS18,Kunjwal19,Kunjwal20}.
	
	\section{Conclusion}
	
	We have uncovered a structural constraint on the intrinsic indeterminism of quantum theory by proving that such indeterminism cannot be extremal for any set of quantum measurements that extract some non-trivial information about the quantum state being measured.
	It seems to be a general feature of quantum theory that extremal correlations consistent with some general physical principle that are achievable by quantum theory are also achievable in its classical deterministic limit, whether it is in the case of Bell scenarios \cite{RTH16} (the principle here being nonsignaling), in contextuality scenarios (\autoref{thm:4}) (the principle here being nondisturbance or consistency of marginals \cite{RSK12,CF12,BCG22}), or even in scenarios with indefinite causal order \cite{KO23} where standard  quantum theory is extended to process matrices \cite{OCB12} with process functions \cite{BW16,BW16fp} as their classical deterministic limit \cite{KO23,KO24,KB23} (only positivity and normalization being assumed in this case). We take our results as further cementing the evidence that this property of quantum correlations---namely, they are convexly extremal in the full space of correlations consistent with some physical principle if and only if they are achievable in the classical deterministic limit---should be taken seriously in any future axiomatization of quantum theory or in theories that aim to extend quantum theory in physically reasonable ways \cite{Hardy01,Hardy05,Barrett07,PPK09,NW10,CDP10, MM11,OCB12,CDP13,FSA13,NGH15,HS15,CSFoils,GKS18}. 
	Furthermore, the possibility of non-trivial quantum realizations of probabilistic models that fail to be quantum models \cite{LSW11,KHF14,KG14,Spekkens14} calls for further investigation of how such probabilistic models might be put to use in quantum information protocols. These protocols could, for example, be built on top of quantum realizations of those joint measurability structures \cite{KHF14, AK20} that yield contextuality scenarios where the set of quantum models coincides with the set of classical models \cite{GKS18}.\footnote{Mathematically, such joint measurability structures correspond to compatibility hypergraphs \cite{KHF14} that are not graphs \cite{HFR14}.} The appropriate notion of nonclassicality to use in such protocols would be that of generalized contextuality, which admits a natural characterization in the language of generalized probabilistic theories \cite{SSW21,SSWS23,SSWSK23,SWS24,SBS24,ZSS25}.
	
	\ack
	R.K.~thanks Tobias Fritz for discussions. This work received support from the French government under the France 2030 investment plan, as part of the Initiative d'Excellence d'Aix-Marseille Université-A*MIDEX, AMX-22-CEI-01.
	
	\appendix

\section{Proofs of key lemmas}\label{app:proofs}
Before we prove \autoref{lem:zero}, we recall and prove the following standard result.
\begin{lemma}\label{lem:tracezero}
	Let $A$ and $B$ be two positive semidefinite matrices. If $\operatorname{Tr}(AB) = 0$, then $AB = 0$.
\end{lemma}

\begin{proof}
	We have
	\begin{align}
		\operatorname{Tr}(AB) &= \operatorname{Tr}(BA) \notag \\
		&= \operatorname{Tr}(B^{1/2} B^{1/2} A) \notag \\
		&= \operatorname{Tr}(B^{1/2} A B^{1/2}) \notag \\
		&= \operatorname{Tr}(B^{1/2} A^{1/2} A^{1/2} B^{1/2}). 
	\end{align}
	
	Let $C = A^{1/2} B^{1/2}$. Then
	\begin{align}
		\operatorname{Tr}(C^{\dagger} C) 
		&= \operatorname{Tr}(B^{1/2} A^{1/2} A^{1/2} B^{1/2})= 0.
	\end{align}
	
	Since $C^{\dagger} C \geq 0$ and $\Tr(C^{\dagger} C) = 0$, it follows that
	\begin{align}
		C^{\dagger} C = 0 \Rightarrow C = 0.
	\end{align}
	
	Hence,
	\begin{align*}
		AB &= A^{1/2} (A^{1/2} B^{1/2}) B^{1/2} = A^{1/2} C B^{1/2} = 0.
	\end{align*}
\end{proof}

\subsection*{Proof of \autoref{lem:zero}}
\begin{proof}
	Since both $\rho$ and $E$ are positive semidefinite, we have from \autoref{lem:tracezero} that 
	$\Tr(\rho E)=0\Rightarrow \rho E=0$. Since we are working in the representation where 
	\begin{align*}
		\rho = \begin{pmatrix}
			\Lambda_r & 0 \\
			0 & 0
		\end{pmatrix},
	\end{align*}
	\begin{align*}
		E = \begin{pmatrix}
			E_{11} & E_{12} \\
			E_{21} & E_{22}
		\end{pmatrix},
	\end{align*}
	and
	\begin{align}
		\rho E &= \begin{pmatrix}
			\Lambda_r E_{11} & \Lambda_r E_{12} \\
			0 & 0
		\end{pmatrix},
	\end{align}
	$\rho E = 0$ gives us the following block-wise equations
	\begin{align}
		\Lambda_r E_{11} = 0, \quad \Lambda_r E_{12} = 0.
	\end{align}
	Since every diagonal entry $\lambda_i > 0$, the diagonal matrix $\Lambda_r$ is invertible. Multiplying the block-wise equations on the left by $\Lambda_r^{-1}$ yields
	\begin{align*}
		E_{11} = 0, \quad E_{12} = 0. 
	\end{align*}
	Because $E$ is Hermitian, $E_{21} = E_{12}^{\dagger} = 0$. Hence,
	\begin{align*}
		E = \begin{pmatrix}
			0 & 0 \\
			0 & E_{22}
		\end{pmatrix},
	\end{align*}
	which completes the proof.
\end{proof}

\subsection*{Proof of \autoref{lem:trivial}}	
\begin{proof}
	The `if' direction is obvious: given $E=\alpha\id$ we have $\Tr(\rho E)=\alpha$ for all $\rho$.
	
	For the `only if' direction, consider the orthonormal basis $\{\ket{i}\}_{i=0}^{d-1}$ of $\mathcal{H}$ that diagonalizes $E$, \textit{i.e.}, 
	\begin{align}
		E=\sum_i\mu_i\ketbra{i},
	\end{align}
	with eigenvalues $\{\mu_i\}_{i=0}^{d-1}$. Suppose, for the sake of contradiction, that $E\neq \alpha\id$. Then there exists some $k\in\{0,1,\dots,d-1\}$ such that $\mu_k\neq\alpha$. Then we have, for $\rho=\ketbra{k}$, that 
	\begin{align}
		\Tr(\ketbra{k}E)=\mu_k\neq\alpha,
	\end{align}
	contradiction our assumption that $\Tr(\rho E)=\alpha$ for all $\rho$. Hence, we must have $E=\alpha\id$.
\end{proof}

\subsection*{Proof of \autoref{lem:unique}}

\begin{proof}
	Any assignment of effects $\{E_v\}_{v\in V(H)}$ to the contextuality scenario $H$ implies that all probabilistic models given by $\{\Tr(\rho E_v)\}_{v\in V(H)}$ are in the set $\mathcal{G}(H)$. 
	
	Uniqueness of the probabilistic model on $H$ implies that all quantum states must yield the same probabilistic model $p$, \textit{i.e.},
	\begin{align}
		\Tr(\rho E_v)=p(v)\quad\forall\rho, \forall v\in V(H).
	\end{align}
	
	In a contextuality scenario with exactly one probabilistic model both of the above conditions must hold. That is, we must have that the effects $\{E_v\}_{v\in V(H)}$ are such that for any choice of quantum state $\rho$ the fixed set of probabilities given by the unique probabilistic model results, \textit{i.e.},
	
	\begin{align}
		\forall \rho: p(v)=\Tr(\rho E_v) \quad \forall v\in V(H).
	\end{align}
	
	Then from \autoref{lem:trivial}, we immediately have that $E_v=p(v)\id$ for all $v\in V(H)$.
\end{proof}

\subsection*{Proof of \autoref{lem:incidence}}	
\begin{proof}
	We prove this by contradiction, \textit{i.e.}, if there exists a non-trivial solution ($x\neq \textbf{0}$), then the contextuality scenario cannot have a unique probabilistic model. 
	
	We have $Ap=\textbf{1}$ for the unique probabilistic model on the contextuality scenario with incidence matrix $A$. Suppose $\exists\, x \neq \textbf{0}$ such that $Ax = \textbf{0}$. Then, for any $c \in \mathbb{R}$,
	\begin{align}
		A(p + cx) = \textbf{1}.
	\end{align}	
	The vector $p + cx$ is a probabilistic model if and only if
	\begin{align}\label{eq:nonunique}
		0 \leq p(v) + c x(v) \leq 1 
		\quad \forall v \in V(H).
	\end{align}
 	
 	Let us first consider the case where the set $S_1:=\{v\in V(H): p(v)=1\}$ is non-empty. From $Ap=\textbf{1}$ and the fact that $p(v)>0$ for all $v\in V(H)$, we have that any $v\in S_1$ appears in exactly one hyperedge $e\in E(H)$ which contains no other vertex $v'\neq v$. Hence, the constraint $A(p + cx)=\textbf{1}$ gives us  
 	\begin{align}
 		1 + c x(v)=1 
 		\quad \forall v \in S_1,
 	\end{align}
 	where $c\in\mathbb{R}$ is arbitrary. In particular, using $c=1$, we have that 
 	\begin{align}
 		&1 + x(v)=1 
 		\quad \forall v \in S_1\\
 		&\Leftrightarrow x(v)=0
 		\quad \forall v \in S_1.
 	\end{align}
 	This, combined with the fact that $p(v)>0$ for all $v\in V(H)$, gives us for any $x$ satisfying $Ax=\textbf{0}$ that
 	 \begin{align}\label{eq:pnot1}
 	 	x(v)\neq 0 \Leftrightarrow 0<p(v)<1.
 	 \end{align}

We can now return to the condition of Eq.~\eqref{eq:nonunique} and rewrite it in the following equivalent form:
	\begin{align}
		-\frac{p(v)}{x(v)} \;\leq\; c \;\leq\; \frac{1-p(v)}{x(v)}, 
		&\quad \text{when } x(v) > 0, \\
		-\frac{p(v)}{x(v)} \;\geq\; c \;\geq\; \frac{1-p(v)}{x(v)}, 
		&\quad \text{when } x(v) < 0,
	\end{align}	
	where, by Eq.~\eqref{eq:pnot1}, we have that $0<p(v)<1$. Combining these bounds, we obtain
	\begin{equation}
		\begin{aligned}
			&\max_{v \in V(H):\, x(v) \neq 0} 
			\Bigg\{ - \min\!\left( \tfrac{1-p(v)}{|x(v)|}, \tfrac{p(v)}{|x(v)|} \right) \Bigg\} \\
			&\hspace{3cm} \leq c \leq \\
			&\min_{v \in V(H):\, x(v)\neq 0} 
			\Bigg\{ \min\!\left( \tfrac{1-p(v)}{|x(v)|}, \tfrac{p(v)}{|x(v)|} \right) \Bigg\}.
		\end{aligned} 
	\end{equation}
	
	Define
	\begin{align}
		c_0 := \min_{v \in V(H):\, x(v) \neq 0} 
		\min\!\left\{\tfrac{1-p(v)}{|x(v)|}, \tfrac{p(v)}{|x(v)|}\right\}. 
	\end{align}
	Clearly $c_0 > 0$. Hence both $p$ and $p + c_0 x$ are distinct and valid probabilistic models, \textit{i.e.},
	\begin{align*}
		Ap = \textbf{1} 
		\quad \text{and} \quad
		A(p + c_0 x) = \textbf{1}.
	\end{align*}
	
	This contradicts uniqueness of $p\in\mathcal{G}(H)$. Therefore, no solution $x\neq \textbf{0}$ exists for the equation $Ax=\textbf{0}$ and we must have $x=\textbf{0}$.
\end{proof}
\vspace{1cm}
\bibliographystyle{apsrev4-2}
\bibliography{masterbibfilev2}
\end{document}